\documentclass[12 pt]{article}
\usepackage[margin=2 cm]{geometry}
\usepackage{amsmath,amsthm,amsfonts,amssymb}

\usepackage{graphicx,latexsym}
\usepackage{lscape}
\usepackage[usenames,dvipsnames]{color}
\def\keyw{\par\medskip\noindent {\it Keywords:}\enspace\ignorespaces}
\newtheorem{thm}{Theorem}[section]

 \newtheorem{lem}[thm]{Lemma}

 \newtheorem{defn}[thm]{Definition}

\numberwithin{equation}{section}

\newtheorem{discu}{Discussion:}

\newtheorem{conje}{Conjecture:}

\begin{document}
%\textcolor{magenta}
%\begin{color}{magenta}
\date{}
%\doublespacing
\title{A note on minimum linear arrangement for BC graphs}
\author{\begin{tabular}{rcl}
         Xiaofang Jiang$^1$, Qinghui Liu$^{1,}$\thanks{This work is partially supported by
         National Natural Science Foundation of China, No. 11371055, No. 11571030.}
         ,
         N. Parthiban$^{2,}$\thanks{This work is partially supported by Project no. SR/S4/MS: 846/13, Department of Science and Technology, SERB, Government of India}~ and R. Sundara Rajan$^3$\\\\
        \end{tabular}\\
        \begin{tabular}{c}
           $^1$Department of Computer Science, Beijing Institute of Technology, \\
           Beijing, China\\
           $^2$School of Advanced Sciences, VIT University, Chennai, India\\
           $^3$Department of Mathematics, Hindustan Institute of Technology and Science, \\
           Chennai, India \\
             \textsl{parthiban24589@gmail.com} \\
        \end{tabular}}
\maketitle
%\normalsize
\vspace{-1 cm}
\begin{abstract}
 A linear arrangement is a labeling or a numbering or a linear ordering of the vertices of a graph. In this paper we solve the minimum linear arrangement problem for bijective connection graphs (for short BC graphs) which include hypercubes, M\"{o}bius cubes, crossed cubes, twisted cubes, locally twisted cube, spined cube, $Z$-cubes, etc. as the subfamilies.
 %Further we study the problem for certain incomplete BC graphs.
\end{abstract}

\vspace{-0.2 cm}
\keyw{Minimum linear arrangement, BC graphs}

\vspace{-0.3 cm}
\section{Introduction}

Graph layout problems are a particular class of combinatorial optimization problems whose goal is to find a linear layout of an input graph in such a way that a certain objective function is optimized. In the literature, there are plenty of layout problems are discussed, such as Linear Arrangement, Bandwidth, Cutwidth, Modified Cut, Sum Cut, Edge Bisection and Vertex Bisection \cite{Petit2011}. A large number of relevant problems in different domains formulated as graph layout problems include VLSI circuit design, network reliability, information retrieval, numerical analysis, computational biology, single machine job scheduling, automatic graph drawing and topology awareness of overlay networks \cite{Diaz2002, Raoufi2013}. The problems are hard in general but known to be solvable in certain restricted classes of graphs \cite{Petit2011}.

A linear arrangement $f$ of an undirected graph $G=(V,E)$ with $n$ nodes is a bijective function $f:V\rightarrow\{1,2,\ldots,n\}$. A linear arrangement is also called a labeling or a numbering or a linear ordering of the vertices of a graph. In \cite{Harper1964}, the Minimum Linear Arrangement (MinLA) problem is formulated as follows: Given a graph $G=(V,E)$, find a linear arrangement $f$ that minimizes $\sum_{(u,v)\in E}|f(u)-f(v)|$. Linear arrangements are a particular case of embedding graphs in $d$-dimensional grids or other graphs. The case in which a graph with $n$ vertices must be embedded into a path $P_{n}$ is perhaps the simplest nontrivial embedding problem. The MinLA problem is NP-complete for  bipartite graphs \cite{Even1975} and permutation graphs \cite{Cohen2006}.

% The graph classes for which problem is polynomial time solvable are incomplete hypercubes  \cite{Miller2015}, Chord graphs \cite{Raoufi2013}, Hypercubes \cite{Harper1964},  $d$-dimensional $c$-ary cliques \cite{Ellis1964}, De Bruijn graph of order 4 \cite{Harper1970}, Rectangular grids, 2-dimensional cylinder  \cite{Muradyan1980}, Square grids  \cite{Mitchison1986}, Trees  \cite{Chung1988}, Complete $p$-partite graphs \cite{Muradyan1988}, Outerplanar graphs  \cite{Frederickson1988}, Certain Halin graphs  \cite{Easton1996}, Proper interval graphs \cite{Safro2002}, Augmented hypercubes \cite{Manuel2011}, Folded hypercubes  \cite{Rajasingh2011}, Circulant graphs  \cite{MaArRaRa13}, Petersen graphs \cite{RaArRaMa11},  Locally twisted cubes  \cite{Arockiaraj2015} and Directed grid graphs  \cite{Rajasingh2015}.

\section{Preliminaries}

The following edge isoperimetric problems are used as tools to solve the MinLA problem. MinLA has been computed for regular graphs such as hypercubes \cite{Harper1964}, circulant graphs \cite{MaArRaRa13}, folded hypercubes \cite{Rajasingh2011}, Petersen graphs \cite{RaArRaMa11} chord graphs \cite{Raoufi2013} and locally twisted cubes \cite{Arockiaraj2015} using edge isoperimetric problem. In this paper, we compute the MinLA for certain families of regular graphs such as BC graphs.
%Recently Miller et al. \cite{Miller2015} computed MinLA for incomplete hypercubes which are not regular. In this paper we extend the result and compute the MinLA for more general incomplete hypercubes which are subfamilies of incomplete BC networks.

\paragraph{Problem 1 :} \cite{Bezrukov1988} For a given $m$, if $\theta_{G}(m)=\underset{A\subseteq V\text{, }\left\vert A\right\vert =m}{\min }\left\vert \theta _{G}(A)\right\vert $ where $\theta _{G}(A)=\{(u,v)\in E:u\in A,v\notin A\}$, then the problem is to find $A\subseteq V$ with $\left\vert A\right\vert =m$ such that $\theta _{G}(m)=\left\vert \theta_{G}(A)\right\vert $.

\paragraph{Problem 2 :} \cite{Bezrukov1988} For a given $m$, if $I_{G}(m)=\underset{A\subseteq V\text{, }\left\vert A\right\vert =m}{\max }\left\vert I_{G}(A)\right\vert $ where $I_{G}(A)=\{(u,v)\in E:u,v\in A\}$, then the problem is to find $A\subseteq V$ with $\left\vert A\right\vert =m$ such that $I_{G}(m)=\left\vert I_{G}(A)\right\vert $. Such a set A is called an optimal set.

%A linear arrangement $f$ of $G$ is said to be an optimal order if $f^{-1}(\{1,2,\ldots,i\})$ is an optimal set
%for any $1\leq i \leq |V(G)|$.

\begin{defn} Let $G$ and $H$ be finite graphs. An \textit{embedding} of $G$ into $H$ is a pair $(f,P_f)$
defined as follows:
\vspace{-0.1 cm}
\begin{enumerate}
\item $f$ is a one-to-one map from $V(G)$ to $V(H)$
\item $P_f$ is a one-to-one map from $E(G)$ to $%
\{P_{f}(u,v):P_{f}(u,v)$ is a path in $H$ between $f(u)$ and $%
f(v)$, for $(u,v)\in E(G)\}.$
\end{enumerate}
\end{defn}

For brevity, we denote the pair $(f,P_f)$ as $f$. The \textit{expansion} of an embedding $f$ is the ratio of the number of vertices of $H$ to the number of vertices of $G$. In this paper, we consider embeddings with expansion one.

The \textit{congestion }of an embedding $f$ of $G$ into $H$ is the
maximum number of edges of the graph $G$ that are embedded on any single
edge $e$ of $H$. Let $EC_{f}(e)$ denote the number of edges $(u,v)$ of $G$
such that $e$ is in the path $P_{f}(u,v)$ between $f(u)$ and $f(v)$ in
$H$. In other words,

\begin{equation*}
EC_{f}(e)=\left\vert \left\{ (u,v)\in E(G):e\in
P_{f}(u,v)\right\} \right\vert
\end{equation*}%
where $P_{f}(u,v)$ denotes the path between $f(u)$ and $f(v)$ in $H$
with respect to $f$.
Further, if $S$ is any subset of $E(H)$, then we define $EC_{f}(S)=\underset{e\in S}{\sum }EC_{f}(e)$.
\begin{defn}
The wirelength of an embedding $f$ of $G$ into $H$
is given by
\vspace{-0.1 cm}
\begin{equation*}
WL_{f}(G,H)=\underset{e\in E(H)}{\sum }EC_{f}(e)
\end{equation*}
\vspace{-0.1 cm}
The \textit{wirelength} of $G$ into $H$ is defined as
\begin{equation*}
WL(G,H)=\min WL_{f}(G,H)
\end{equation*}%
where the minimum is taken over all embeddings $f$ of $G$ into $H$.
\end{defn}

When $H$ is a path, we represent $WL_f(G,H)$ by $LA_f(G)$ and represent $WL(G,H)$ by \textit{MinLA}$(G)$.

\begin{lem}\label{lem1}
The MinLA of a graph $G$ of order $n$ is given by
$$MinLA(G) \geq ~\overset{n-1}{\underset{i=1}\sum}~\theta_G(i).$$
\end{lem}
\noindent \textit{Proof.}
For $1\le i<n$, let $S_i={(i,i+1)}$, then for any embedding $f$, we have
%$$EC_f(S_i)\ge \theta_G(i)$$
\begin{eqnarray*}
  EC_f(S_i) &\geq & \theta_G(i) \\
  \therefore ~{\underset{f}\min} ~\overset{n-1}{\underset{i=1} \sum}~EC_f(S_i) &\geq & \overset{n-1}{\underset{i=1}\sum}~\theta_G(i) \\
  i.e., MinLA(G) &\geq & \overset{n-1}{\underset{i=1}\sum}~\theta_G(i). ~~~~\square
\end{eqnarray*}
%
%
%\textcolor{magenta}{Parthiban: As you suggest help us to write Proof for above lemma}
%
%\textcolor{magenta}{Qh: I only write a framework of the proof. You can modify it.}

\section{Main Results}

BC networks have received a great deal of attention in the past \cite{Fan2003, Fan2005, Fan2008, Tan2008}. Fan et al. \cite{Fan2003} proposed a family of interconnection networks called BC graphs. BC networks are a class of networks which include several well-known interconnection networks like hypercubes, M\"{o}bius cubes \cite{Shawn1995}, crossed cubes \cite{Efe1991}, twisted cubes \cite{Abraham1991}, locally twisted cube \cite{YaEvMe05}, spined cube \cite{Zhou2011} and Z-cube \cite{Zhu2015}. These variations of hypercubes generally possess certain superior properties over the hypercubes and are recognized as attractive alternatives to the hypercubes.

%In the literature, embedding problems such as Bandwidth \cite{Ha66}, Cutwidth \cite{Ha64}, MinLA \cite{AdHu73}, Cyclic bandwidth \cite{LeVoWi84}, Cyclic cutwidth \cite{ChTr98} and Cyclic wirelength \cite{BeSc98} are studied with variants of hypercubes as host graph. Since hypercube and certain variants of hypercubes are subclass of BC graphs, it is important to study the embedding problems onto BC graphs. To define an $n$-dimensional BC networks $X_n$, we first define a bijective connection as follows:

%\textcolor{magenta}{Parthiban: Please check mathematical definition of BC graphs}
%
%\textcolor{magenta}{Qh: I think the definition of BC graphs in reference [28] is better.}

\begin{defn}\rm{\cite{Fan2003,Tan2008}}
Let $G_1=(V_1,E_1)$, $G_2=(V_2,E_2)$ be two vertex disjoint graph of the same order.
A bijective connection between $G_1$ and $G_2$ is defined as an edge set
$E=\{(v,\phi(v))\}$, where $\phi: V_1\rightarrow V_2$ is a bijection.
Define $G_1\oplus G_2=(V_1\cup V_2, E_1\cup E_2\cup E)$.
\end{defn}

An $n$-dimensional BC graph, denoted by $X_n$, is an $n$-regular graph with $2^n$ nodes and $n~2^{n-1}$ edges. The set of all the $n$-dimensional BC graphs is called the family of the $n$-dimensional BC graphs,
denoted by $ \ell_n$. We now define $X_n$ mathematically as follows:

\begin{defn} \rm{\cite{Fan2003,Tan2008}} The one-dimensional BC network $X_1$ is a complete graph with two vertices, $K_2$.
The family of the one-dimensional BC network is defined as $\ell_1=\{K_2\}$.
When $n\ge2$, $G= X_n \in \ell_n$ if and only if $G=G_1\oplus G_2$ for some $G_1,G_2\in \ell_{n-1}$.
\end{defn}

%\vspace{0.5 cm}
%\textbf{To Parthiban: Write a proper definition of BC graphs.}

\begin{lem}  \rm{\cite{Tan2008}}
\label{lemma1} Let $G$ be a $n$-dimensional BC graph.  For an integer $m$, which can be uniquely written as $ m = \sum\limits_{i=1}^{r-1} {2^{l_i}}$ for some nonnegative integers $r$ and $l_0 > l_1 > \ldots >l_{r-1}$, then the maximum number of edges joining vertices from a set of $m$ vertices is $I_G(m) = \sum\limits_{i=0}^{r-1} {(l_i/2+i)2^{l_i}}$, where $1 \leq m \leq 2^n$, $n \geq 1$.
\end{lem}

Note that this implies $\theta_G(m)=nm-2I_G(m)=\sum\limits_{i=0}^{r-1} {(n-l_i-2i)2^{l_i}}.$ And hence
$$\sum\limits_{m=1}^{2^n-1}\theta_{G}(m)=
\sum\limits_{k=0}^{n-1}2^k
\sum\limits_{i=0}^{n-k-1}\left((n-k-2i)2^k\left(\begin{array}{c}n-k-1\\ i\end{array}\right)
\right),$$
where for a sequence $n>l_0 > l_1 > \ldots >l_{r-1}\ge0$ with $l_i=k$ for some $0\le i<n$, 
there are $2^k$ choices for $l_{i+1},\ldots, l_{r-1}$
and $\left(\begin{array}{c}n-k-1\\ i\end{array}\right)$ choices for $l_{0},\ldots, l_{i-1}$. Then

$$\begin{array}{rcl}
\sum\limits_{m=1}^{2^n-1}\theta_{G}(m)&=&
\sum\limits_{k=0}^{n-1}2^k\left(
2^{n-1}(n-k)-2^{k+1}(n-k-1)\sum\limits_{i=1}^{n-k-1}\left(\begin{array}{c}n-k-2\\ i-1\end{array}\right)\right)\\
&=&\sum\limits_{k=0}^{n-1}2^k\left(
2^{n-1}(n-k)-2^{k+1}(n-k-1)2^{n-k-2}\right)\\
&=&\sum\limits_{k=0}^{n-1}2^{n-1+k}\ =\ 2^{n-1}(2^n-1).
\end{array}
$$

Now, we define a class of linear arrangement from BC graph to path by induction.

\begin{defn}
For any $n\ge0$, define $P_{2^n}=\{1,2,\cdots, 2^n\}\subset\mathbb{Z}$.
Let $G=X_n$ be the $n$-dimensional BC graph $G$.
If $n=1$, writing $V(G)=\{v_1,v_2\}$, we define $f(v_1)=1$, $f(v_2)=2$.
We call it a $1$-dimensional BC structure linear arrangement.
If $n>1$, then there exist $(n-1)$-dimensional BC graphs $G_1$, $G_2$ such that $G=G_1\oplus G_2$.
Let $f_1: V(G_1)\rightarrow P_{2^{n-1}}$, $f_2: V(G_2)\rightarrow P_{2^{n-1}}$
be two $(n-1)$-dimensional BC structure linear arrangements.
Define $f: V(G)\rightarrow P_{2^n}$ as follows. For any $v\in V(G_1)$, let $f(v)=f_1(v)$ and
for any $v\in V(G_2)$, let $f(v)=f_2(v)+2^{n-1}$. We call it a $n$-dimensional BC structure linear arrangement.
\end{defn}
%\subsection{MinLA of BC Networks}
%The next theorem is immediately follows from Lemma \ref{lemma1} and Theorem \ref{thm1}.

\begin{thm}\label{thm3}
The MinLA of BC graph $X_n$ is $$MinLA(X_n)=\sum_{i=1}^{2^n-1}\theta_{X_n}(i)=2^{n-1}(2^n-1).~~~\square$$
\end{thm}

%\textcolor{magenta}{Parthiban: Kindly help us to write the proof for $MinLA(X_n)$ = $2^{n-1}(2^n-1)$ using lemma 3.3, lemma 2.5.}
%
%\textcolor{magenta}{Qh: I remember that some reference proved that there exist $f$ such that
%$f^{-1}(\{1,\cdots,i\})\}$ for any $1\le i\le 2^n$ is 'cubal' or 'incomplete cube';
%and if $X_n=G_1\oplus G_2$, then $1\le f(G_1)\le 2^{n-1}$, $2^{n-1}+1\le f(G_2)\le 2^{n}$.
%So please add corresponding result and reference here. We call this kind of linear arrangements 'BC linear arrangements'.}

\begin{proof}
By Lemma \ref{lem1} and analysis above,
we have $MinLA(X_n)\geq \sum_{i=1}^{2^n-1}\theta_{X_n}(i)=2^{n-1}(2^n-1)$. To prove the equality, we need to show that for any $X_n\in\ell_n$, for any BC structure linear arrangement $f_n$,
$$LA_{f_n}(X_n)=2^{n-1}(2^n-1).$$
It is direct to show that $LA_{f_1}(X_1)=1=2^0(2^1-1).$

Suppose for $n\ge 2$, any $X_{n-1}\in\ell_{n-1}$ and any BC structure linear arrangement $f_{n-1}$,
$$LA_{f_{n-1}}(X_{n-1})=2^{n-2}(2^{n-1}-1).$$

Take any $X_{n}\in\ell_{n}$ and any $n$-dimensional BC linear arrangement $f_{n}: V(X_n)\rightarrow P_n$.
Then there is $G_1,G_2\in \ell_{n-1}$ such that $X_n=G_1\oplus G_2$ and there are
$(n-1)$-dimensional BC structure linear arrangements
$f': V(G_1)\rightarrow P_{n-1}$, $f'': V(G_2)\rightarrow P_{n-1}$ such that
for any $v\in V(G_1)$, $f_n(v)=f'(v)$,
for any $v\in V(G_2)$, $f_n(v)=f''(v)+2^{n-1}$.
Then
$$LA_{f_n}(X_n) = LA_{f'}(G_1)+LA_{f''}(G_2)+\sum_{a\in G_1,\ b\in G_2}|P_{f_n}(a,b)|.$$

%\begin{eqnarray*}
%% \nonumber to remove numbering (before each equation)
%  LA_{f_n}(X_n) &=& LA_{f'}(G_1)+LA_{f''}(G_2)+\sum_{a\in G_1,\ b\in G_2}|P_{f_n}(a,b)| \\
%   &=& \sum_{i=1}^{2^{n-1}-1}\theta_{X_{n-1}}(i)+\sum_{i=1}^{2^{n-1}-1}\theta_{X_{n-1}}(i)+\sum_{a\in G_1,\ b\in G_2}|P_{f_n}(a,b)|.
%\end{eqnarray*}

By induction hypothesis, $LA_{f'}(G_1)=LA_{f''}(G_2)=2^{n-2}(2^{n-1}-1)$.

By direct computation,
$$\sum_{a\in G_1,\ b\in G_2}|P_{f_n}(a,b)|=2^{n-1}+2\sum_{i=1}^{2^{n-1}-1}i=2^{2n-2}.$$
Thus for any $n$-dimensional BC structure linear arrangement $f_n$,
we have $$LA_{f_n}(X_n)=2^{n-1}(2^n-1).$$.
Hence, the theorem is proved by induction.
\end{proof}

%\subsection{MinLA of Incomplete BC Networks}
%
%In a large interconnection network, nodes or edges often develop faults. It is very important to study graph embedding in the case where some nodes or edges in the graphs have become faulty \cite{FaLiPaJi07}. It is desirable to find an embedding of a guest graph into a host graph where all faulty nodes and edges have been removed.
%
%\textcolor{magenta}{Kindly check definition.}
%
%\begin{defn}
%The incomplete BC graph $G$ is obtained from $X_i$ and $X_j$, $i>j$ and $i,j\geq 2$ such that the induced subgraph $G[X_{i^*},X_j]$ is isomorphic to $X_{j+1}$, where $X_{i^*}$ is a subgraph of $X_i$ and $i^*=j$.
%\end{defn}
%
%
%\begin{rem}
%The graph $G$ contains $2^i+2^j$ vertices and $i~2^{i-1}+2^{j-1}(j+2)$ edges.
%\end{rem}

\section{Concluding Remarks}

In this paper, we computed the MinLA for BC graphs. Finding the other parameters, such as bandwidth, cutwidth, edge bisection and vertex bisection for BC graphs are under investigation.

\end{document}